\documentclass[11pt,a4paper]{paper}

\usepackage{graphicx}
\usepackage{enumerate}
\usepackage{amsmath,amsthm,amsfonts}
\usepackage{microtype}
\usepackage{subcaption}
\usepackage{wrapfig}
\usepackage[a4paper]{geometry}
\usepackage[T1]{fontenc}

\usepackage{lmodern}

\graphicspath{{.}}

\usepackage[usenames,dvipsnames]{xcolor}
\usepackage[noend,linesnumbered,boxed]{algorithm2e}
\SetAlCapSkip{0.2cm}

\usepackage{cite}

\newtheorem{theorem}{Theorem}[section]

\newtheorem{lemma}[theorem]{Lemma}

\newcommand{\mathset}[1]{\ensuremath {\mathbb {#1}}}

\newcommand{\R}{\mathset{R}}

\newcommand\eps{\varepsilon}
\newcommand\etal{et al.\xspace}

\DeclareMathOperator{\diam}{diam}

\DeclareMathOperator{\UD}{UD}

\usepackage{ifdraft}

\title{Routing in Unit Disk
Graphs\footnote{This work is supported in part by GIF
  project 1161 \&  DFG projects MU/3501/1. A
  preliminary version appeared as
  Haim Kaplan, Wolfgang Mulzer, Liam Roditty, and Paul Seiferth.
\emph{Routing in Unit Disk Graphs.}
Proc.~12th LATIN, 2016.}}

\author{Haim Kaplan\thanks{School of Computer Science, Tel Aviv University,
 Israel, \texttt{haimk@post.tau.ac.il}} \and
 Wolfgang Mulzer\thanks{Institut f\"ur Informatik,
Freie Universit\"at Berlin,
  Germany
  \texttt{\{mulzer,pseiferth\}@inf.fu-berlin.de}} \and
Liam Roditty\thanks{Department of Computer Science, Bar Ilan University,
  Israel
  \texttt{liamr@macs.biu.ac.il}} \and
Paul Seiferth\footnotemark[3]}

\begin{document}
\maketitle

\begin{abstract}
Let $S \subset \R^2$ be a set of $n$ sites. The unit disk graph
$\UD(S)$ on $S$ has
vertex set $S$ and an edge between two distinct sites $s,t \in S$ if and
only if $s$ and $t$ have Euclidean distance $|st| \leq 1$.

A routing scheme $R$ for $\UD(S)$ assigns to each site $s \in S$ a
\emph{label} $\ell(s)$ and a \emph{routing table} $\rho(s)$.
For any two sites $s, t \in S$,
the scheme $R$ must be able to route a packet from
$s$ to $t$ in the following way: given a \emph{current site} $r$
(initially, $r = s$),
a \emph{header} $h$ (initially empty), and
the \emph{label} $\ell(t)$ of the target,
the scheme $R$ consults the  routing table
$\rho(r)$ to compute a
neighbor $r'$ of $r$, a new header $h'$, and the label $\ell(t')$
 of an intermediate target $t'$. (The label of the original target may be stored
at the header $h'$.)
The packet is then routed to $r'$,
and the procedure is repeated until the packet reaches $t$.
The resulting sequence of sites is called the
\emph{routing path}.
The \emph{stretch} of $R$ is the maximum ratio of the (Euclidean)
length of the routing path produced by $R$ and the shortest path in
$\UD(S)$, over all pairs of distinct sites in $S$.

For any given $\eps > 0$,
we show how to construct a routing scheme for $\UD(S)$ with stretch
$1+\eps$ using labels of $O(\log n)$ bits and routing tables of
$O(\eps^{-5}\log^2 n \log^2 D)$ bits, where $D$ is the (Euclidean)
diameter of
$\UD(S)$. The header size is $O(\log n \log D)$ bits.
\end{abstract}

\section{Introduction}
Routing in graphs constitutes a fundamental problem in
distributed graph algorithms~\cite{GiordanoSt04,PelegUp89}.
Given a graph $G$, we would like to be able to
route a packet from any node in $G$ to any other node,
where the destination node is represented by its \emph{label}.
The routing algorithm should be \emph{local}, meaning that it
uses only information stored with the packet and
with the current node, and it should be \emph{efficient}, meaning
that the packet does not travel much longer than necessary.
There is an obvious solution to this problem:
with each node $s$ of $G$, we store the shortest path tree
for $s$. Then it is easy to route a packet along the shortest
path to its destination. However, this solution is very inefficient:
we need to store the complete topology of $G$ with each node, leading
to quadratic space usage. Thus, the goal of a routing scheme
is to store as little information as possible with each node
of the graph, such that we can still route a packet on a path
of length close to shortest.

For general graphs a plethora of results is available, reflecting
the work of almost three decades (see, e.g., \cite{RodittyTo15,Chechik13} and
the references therein).
However, for general graphs, any efficient routing scheme needs
to store $\Omega(n^{\alpha})$ bits per node, for some
$\alpha > 0$~\cite{PelegUp89}.
Thus, it is natural to ask whether improved results are possible
for specialized graph classes.
For example, for trees it is known how to obtain a routing scheme that
follows a shortest path and requires $O(\log n)$ bits of information at
each node~\cite{ThorupZw01,FraigniaudGa01,SantoroKhatibK85}.
In planar graphs, for any $\eps > 0$ it is possible to store a
polylogarithmic number of  bits at
each node in order to route a packet along a path of length at most
$1+\eps$ times the length of the shortest path~\cite{Thorup04}.

A graph class that is of particular interest for routing problems
comes from the study of mobile and wireless networks. Such networks
are traditionally modeled as \emph{unit disk graphs}~\cite{ClarkCo90}.
The nodes are represented by points in the plane, and two nodes are connected
if and only if the distance between the corresponding points is at most
one.\footnote{Alternatively, a unit disk graph is the
intersection graph of a set of disks of radii $1/2$.}
Even though unit disk graphs may be dense,
they share many properties with planar graphs, in particular with respect
to algorithmic problems.
There exists a vast literature on routing in unit disk graphs,
developed in the wireless networking community
(cf.~\cite{GiordanoSt04}). Most of these schemes were designed
with a different outlook, aiming at practical and simple solutions instead
of provable worst-case guarantees. For example, the most popular
routing method is called \emph{geographic routing}. Here, we assume
that the coordinates of the target site are known, and we route
the packet to the neighbor that is closest to the target. Even
though this is a good practical heuristic, it can happen that
a packet gets stuck. There are several ways to modify
geographic routing to obtain a routing scheme where a packet
always reaches its target
(if possible)~\cite{BoseMoStUr01,KarpKu00,KuhnWaZhZo03}. However,
in these schemes the routing path may be much longer than
the shortest path.

To the best of our knowledge, the only compact routing scheme
for unit disk graphs that achieves routing paths that are provably within
a constant factor of the optimum is due to Yan, Xiang, and
Dragan~\cite{YanXiDr12}.
More precisely, they show how to assign a label of
$O(\log^2 n)$ bits to each node of the graph such that
given the labels of a source $s$ and of a target $t$, one
can compute a neighbor of $s$ that leads toward $t$.
They prove that by repeating this procedure, one can obtain a path
from $s$ to $t$ with at most $3d_h(s,t) + 12$ hops, where $d_h(s,t)$
is the hop distance of $s$ and $t$.
For this, Yan~\etal extend a scheme by
Gupta~\etal~\cite{GuptaKuRa04} for planar graphs to unit disk graphs by
using a delicate
planarization argument
to obtain small-sized balanced separators. Even though the scheme by
Yan~\etal is conceptually
simple, it requires a detailed analysis with an extensive case
distinction.

We propose an alternative approach to routing in unit disk graphs.
Our scheme is based on the well-separated pair decomposition
for unit disk graphs~\cite{GaoZh05}. It stores
a polylogarithmic number of bits at each node of the graph,
and it constructs a routing path that can be made arbitrarily close
to a shortest path, where the edges are weighted according
to their Euclidean length
(see Section~\ref{sec:prelims} for a precise statement
of our results). This compares favorably with the scheme by
Yan \etal~\cite{YanXiDr12} which achieves only a constant factor
approximation. Furthermore, our labels need only $O(\log n)$ bits and
our scheme is arguably simpler to analyze.
However, unlike the algorithm by Yan \etal, our scheme requires that the
packet contain a modifiable \emph{header} with a polylogarithmic number of
bits. It is an interesting open question whether a scheme
with similar performance guarantees that does not require
a modifiable header exists.

\section{The Model and Our Results}
\label{sec:prelims}

Let $S \subset \R^2$ be a set of $n$ \emph{sites} in the plane.
We say that $S$ has \emph{density} at most $\vartheta$ if every unit disk
contains at most $\vartheta$ points from $S$. The density
$\vartheta$ of $S$ is
\emph{bounded} if $\vartheta = O(1)$.
The \emph{unit disk graph} for $S$ is the graph
$\UD(S)$
with vertex set $S$ and an edge $st$ between two
distinct sites $s, t \in S$
if and only if $|st| \leq 1$, where $|\cdot|$ denotes the Euclidean
distance.  We define the \emph{weight} of the edge
$st$ to be its Euclidean length and use $d(\cdot, \cdot)$
to denote the shortest path distance in $\UD(S)$.
Given a set $T \subset S$, we define $\diam(T)$ as
the (Euclidean) \emph{diameter} of the induced subgraph
$\UD(T)$ of $\UD(S)$, i.e., the maximum length of a shortest
path between two sites in $\UD(T)$.

We would like to compute a \emph{routing scheme} for $\UD(S)$ with a
small \emph{stretch} and compact \emph{routing tables}.
Formally, a \emph{routing scheme} for $\UD(S)$ consists of
(i) a \emph{label} $\ell(s) \in \{0, 1\}^*$
and (ii) a \emph{routing table} $\rho(s) \in \{0, 1\}^*$, for each site $s \in S$.
The labels and the routing tables correspond to
 a \emph{routing function}
$f : S \times \{0, 1\}^* \times \{0, 1\}^*
\rightarrow S \times \{0,1\}^*\times \{0, 1\}^*$. The  function
$f$ takes as input
a \emph{current site} $s$,
the \emph{label} $\ell(t)$ of a \emph{target site} $t$,  and a \emph{header}
$h \in \{0, 1\}^*$.
The routing function may use its input and the routing table
$\rho(s)$ of $s$ to compute a new site $s'$,
a modified
header $h'$, and the label of an \emph{intermediate target}
$\ell(t')$. The new site $s'$ may be either $s$ or a neighbor
of $s$ in $\UD(S)$. If $s'=s$ then the packet stays at $s$ and
we recompute the routing function at $s$ with the modified
header and the label of the intermediate target.
If $s'$ is a neighbor of $s$ then $s$ sends the packet (with
the header and the label of the intermediate target) to $s'$.
Even though the eventual goal of the packet is the target $t$,
we introduce the intermediate target
$t'$ into the notation, since it allows for a more succinct
presentation of the routing algorithm. (The original target can
be stored with the modifiable header and will be
extracted later according to the definition of the routing function.
Similarly, the routing may proceed through several intermediate
targets, but at each point in time, the routing function
receives only one of them.)

Let $h_0$ be the empty header.
For any two sites $s, t  \in S$,
consider the sequence of triples given by
$(s_0, \ell_0, h_0) = (s, \ell(t), h_0)$
and $(s_{i}, \ell_i, h_i) = f(s_{i-1}, \ell_{i-1}, h_{i-1})$
for $i \geq 1$. We say that the
 routing scheme is \emph{correct} if for any two sites
 $s, t  \in S$ there exists an index $k$ such that $s_k = t$.

We consider the minimal such index $k$,
 $k = k(s, t) \geq 0$ such that $s_k = t$ and
$s_i \neq t$ for $i < k$. We say that \emph{the
routing scheme reaches $t$ after $k$ steps}.
We call $s_0, s_1, \dots, s_k$ the \emph{routing path} between
$s$ and $t$, and we define the \emph{routing distance}
$d_\rho(s, t)$ between $s$ and $t$
as $d_\rho(s, t) = \sum_{i = 1}^{k} |s_{i-1}s_{i}|$.
Recall that $|\cdot|$ denotes the Euclidean distance.

The quality of the routing scheme is measured by several parameters:
\begin{itemize}
 \item the \emph{label size} $L(n) = \max_{|S|= n}
\max_{s \in S} |\ell(s)|$,

\item
the \emph{table size}
$T(n) = \max_{|S|= n}
\max_{s \in S} |\rho(s)|$,

\item
the \emph{header size}
$H(n) = \max_{|S|= n}
\max_{s \neq t \in S}\max_{i = 1, \dots, k(s, t)} |h_i|$,

\item and the \emph{stretch}
$\varphi(n) = \max_{|S|= n}\max_{s \neq t \in S} \frac{d_\rho(s, t)}{d(s, t)}$.
\end{itemize}

We show that for any $S \subset \R^2$, $|S| = n$,
and any $\eps > 0$ we can construct a routing scheme
with  $\varphi(n)= 1+\eps$, $L(n) = O(\log n)$,
$T(n) = O(\eps^{-5}\log^2 n \log^2 D)$, and  $H(n) = O(\log n \log D)$,
where $D = \diam(S)$ is the diameter of $\UD(S)$.
We emphasize that in a unit disk graph, we always have
$D \leq n$. We may also assume that $D \geq 2$: otherwise,
$S$ could be approximated by an $\eps$-net with $O(1)$ vertices,
and we could route in $\UD(S)$
with routing tables and headers of constant size
(see Section~\ref{sec:unboundeddensity}).

The high dependence on $1/\eps$ renders
our result mostly of theoretical interest. However,
it demonstrates for the first time that the stretch
can be made arbitrarily close to $1$ while maintaining
routing tables of polylogarithmic size.

Even though our algorithm uses ideas from previous
routing schemes, such as ``interval routing'' or
hierarchical clustering~\cite{RodittyTo15}, to the
best of our knowledge, we are the first to use the
well-separated pair decomposition~\cite{CallahanKo95}
as the basis for a routing scheme. In general graphs,
this approach is not possible, since in this metric, small
well-separated pair decompositions do not exist~\cite{MitchellMu17}.
As discussed above, in geometric settings, the focus has been
on position-based methods that make stronger use of the
geometry than the WSPD does~\cite{GiordanoSt04}. Thus,
our approach is more geometric than the traditional methods
for general graphs, and at the same time more combinatorial
than the methods based on geometry.

\section{The Well-Separated Pair Decomposition for $\UD(S)$}
\label{sec:wspd}

Our routing scheme uses the well-separated pair decomposition
(WSPD) for
the unit disk graph metric given by Gao and Zhang~\cite{GaoZh05}.
WSPDs provide a compact way to efficiently
encode the approximate pairwise distances in a metric space.
Originally, WSPDs were introduced by Callahan and
Kosaraju~\cite{CallahanKo95}
in the context of the Euclidean metric, and they have found numerous
applications since then (see e.g., \cite{GaoZh05,NarasimhanSmid07}
and the references therein).

Since our routing scheme relies crucially on the specific structure
of the WSPD described by Gao and Zhang, we remind the reader of the main steps
of their algorithm and analysis.

First, Gao and Zhang assume that $S$ has bounded density
and that $\UD(S)$ is connected.
They construct the Euclidean minimum spanning tree $T$ for $S$.
It is well known that $T$ is a spanning tree for $\UD(S)$ with
maximum degree $6$. Furthermore, $T$ can be constructed
in $O(n \log n)$ time~\cite{4M}.
Since $T$ has maximum degree $6$, there exists an edge $e$ in $T$
such that $T \setminus e$ consists of two trees with at least
$\lceil (n-1)/6 \rceil$ vertices each. By applying this observation
recursively, we obtain a
\emph{hierarchical decomposition} $H$ of $T$. The decomposition
$H$ is a binary tree. Each node $v$ of $H$ represents
a subtree $T_v$ of $T$ with vertex set $S_v \subseteq S$ such that (i) the
root of $H$ corresponds
to $T$; (ii) the leaves of $H$ are in one-to-one correspondence with
the sites in $S$; and
(iii) let $v$ be an inner node of $H$ with children $u$ and $w$. Then
$v$ has an \emph{associated edge}
$e_v \in T_v$ such that removing $e_v$ from $T_v$ yields the
two subtrees $T_u$ and $T_w$ represented by $u$ and $w$.
Furthermore, we have $|S_u|, |S_w| \geq \lceil(|S_v| - 1)/6\rceil$.

\begin{figure}[htb]
\centering
 \includegraphics[scale=0.7]{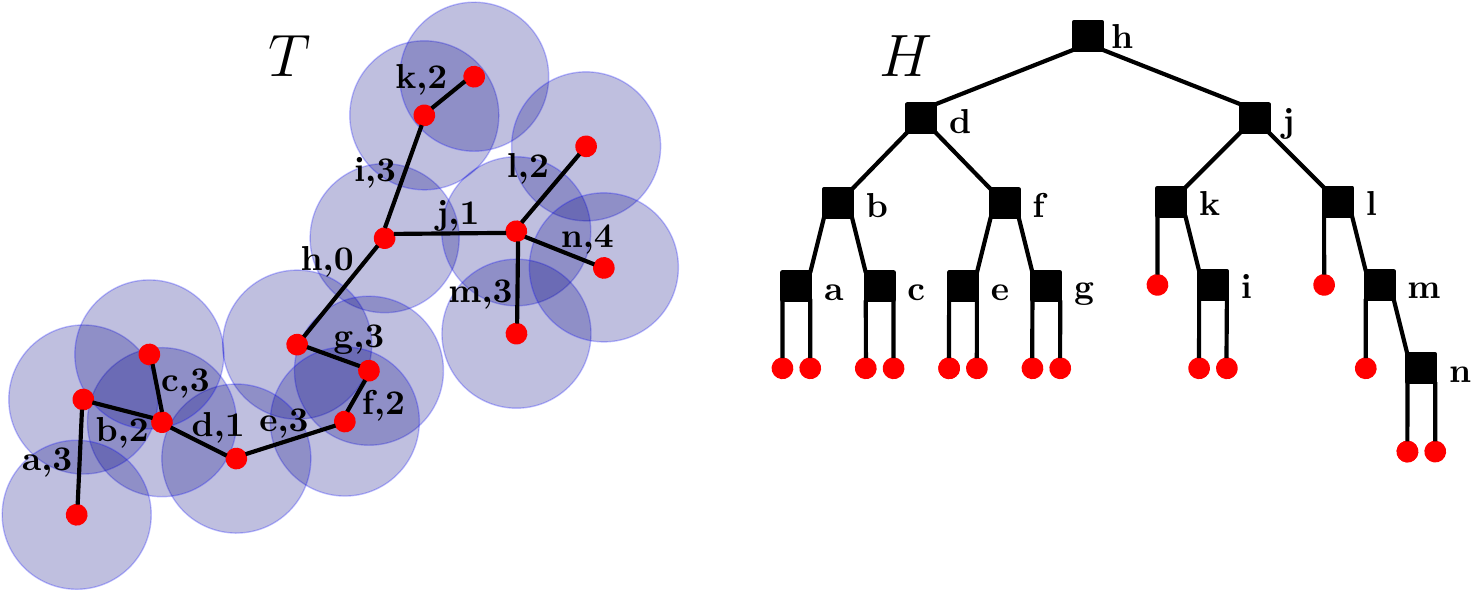}
\caption{An EMST $T$ of $\UD(S)$ (left) where the edges are annotated
with their level in the hierarchical decomposition $H$ (right).}
\label{fig:decomposition}
\end{figure}

It follows that $H$ has height $O(\log n)$.
The \emph{depth} $\delta(v)$ of a node $v \in H$ is defined as the number
of edges on the path from $v$ to the root of $H$.
The \emph{level} of the associated edge $e_v$ of $v$ is
the depth of $v$ in $H$. This uniquely defines a level for
each edge in $T$. Now, for each node $v \in H$, the subtree $T_v$
is a connected component in the forest that is induced in $T$
by the edges of level at least $\delta(v)$
(see Figure~\ref{fig:decomposition}).

After computing the hierarchical decomposition,
the algorithm of Gao and Zhang essentially uses the greedy
method of Callahan and Kosaraju~\cite{CallahanKo95} to construct a WSPD, with $H$ instead
 of the quadtree (or the fair split tree) in \cite{CallahanKo95}.
Let $c \geq 1$ be a separation parameter.
The algorithm traverses $H$ and produces a sequence
$\Xi = (u_1, v_1), (u_2, v_2), \dots, (u_m, v_m)$ of
pairs of nodes of $H$, with the following properties:
\begin{enumerate}
\item
The sets $S_{u_1} \times S_{v_1},
S_{u_2} \times S_{v_2}, \dots, S_{u_m} \times S_{v_m}$
constitute a partition of $S \times S$. This means that
for each ordered pair of sites $(s, t) \in S \times S$, there is
exactly one pair $(u, v) \in \Xi$
with $(s, t) \in S_{u} \times S_{v}$. We say that
$(u, v)$ \emph{represents} $(s, t)$.
\item
Each pair $(u, v) \in \Xi$ is $c$-\emph{well-separated}, i.e.,
we have
\begin{equation}
\label{equ:well-separated}
(c+2)\max\{|S_{u}| - 1, |S_{v}| - 1\} \leq |\sigma(u)\sigma(v)|,
\end{equation}
where $\sigma(u), \sigma(v)$ are  sites in $S_u$ and
$S_v$, respectively, chosen by the algorithm. (This, in fact, holds for any pair of sites in
$S_u$ and
$S_v$, since the algorithm chooses $\sigma(u)$ and $\sigma(v)$ arbitrarily.)
\end{enumerate}
Since in the shortest path metric of $\UD(S)$ the diameter
$\diam(S_u)$ is at most $|S_u| - 1$ and since
$|\sigma(u)\sigma(v)|
\leq d(\sigma(u), \sigma(v))$,
(\ref{equ:well-separated}) implies that
\begin{equation}
\label{equ:wspd_traditional}
(c+2)\max\{\diam(S_{u}), \diam(S_{v})\} \leq
d(\sigma(u), \sigma(v)),
\end{equation}
which is the traditional well-separation condition. However,
(\ref{equ:well-separated}) is easier to check algorithmically and
has additional advantages that we will exploit in our routing scheme
below.

Gao and Zhang show that their
algorithm produces a $c$-WSPD with $m = O(\vartheta c^2 n \log n)$ pairs,
where $\vartheta$ is the density of $S$.
More precisely, they prove the following lemma:
\begin{lemma}[Lemma 4.3 and Corollary 4.6 in \cite{GaoZh05}]
\label{lem:wspdpairs}
For each node $u \in H$, the WSPD $\Xi$ has
$O(\vartheta c^2|S_u|)$ pairs that contain $u$.\qed
\end{lemma}

\section{Preliminary Lemmas}

We begin with two technical lemmas on WSPDs
that will be useful later on.
The first lemma shows that the choice of
the sites $\sigma(u)$ for the nodes
$u \in H$ is essentially arbitrary.

\begin{lemma}
\label{lem:diameter}
Let $\Xi$ be a $c$-WSPD for $S$ and
let $s, t$ be two sites such that the pair $(u,v) \in \Xi$
represents $(s,t)$.
Then $c\diam(S_u) \leq c (|S_u| - 1)\leq d(s, t)$.
\end{lemma}
\begin{proof}
 By triangle inequality and (\ref{equ:well-separated}) we have
\begin{align*}
 |st| & \geq |\sigma(u)\sigma(v)| -
 2 \max\{\diam(S_u), \diam(S_v)  \} \\
 & \geq (c+2) \max\{|S_u| - 1, |S_v| - 1\} -
 2 \max\{\diam(S_u), \diam(S_v)  \}.
\end{align*}
Since $|S_u| -1$ and $|S_v| - 1$ are upper bounds for
$\diam(S_u)$ and
$\diam(S_v)$, respectively, and since
$d(s,t) \geq |st|$,  the claim follows.
\end{proof}

The next lemma shows that short distances are represented
by singletons in the WSPD.

\begin{lemma}
\label{lem:closepairs}
Let $\Xi$ be a $c$-WSPD for $S$ and let $s, t \in S$ be two sites
with $d(s, t) < c$.
If $(u, v) \in \Xi$ represents $(s,t)$, then
$S_u = \{s\}$ and $S_v = \{t\}$.
\end{lemma}
\begin{proof}
If $d(s,t) < c$, by Lemma~\ref{lem:diameter} we have
\begin{equation*}
c (|S_u| - 1)\leq d(s, t) < c,
\end{equation*}
and thus $|S_u| < 2$. The argument for $|S_v|$ is analogous.
\end{proof}

\section{The Routing Scheme}
Let $\vartheta$ be the density of $S$. First we describe a routing scheme
whose parameters depend on $\vartheta$. Then we show how to remove this dependency
and extend the scheme to work with arbitrary density.
Our routing scheme uses the WSPD
described in Section~\ref{sec:wspd}, and it is based on the following idea:
let $\Xi$ be the $c$-WSPD for $\UD(S)$ and let $T$ be the EMST for $S$
used to compute it.
We distribute the information about
the pairs in $\Xi$ among the sites in $S$ (in a way to be described
later) such that each
site stores $O(\vartheta c^2\log n)$ pairs in its routing table.

To route from
$s$ to $t$, we explore $T$, starting from
$s$, until we find the site $r$ with the pair $(u,v)$
representing $(s,t)$. Our scheme will guarantee that
$s$ and $r$ are sites in $S_u$, and therefore it suffices to walk
along $T_u$
to find $r$ (see Figure~\ref{fig:recursive}).
We call this first step in which we search for $(u,v)$
the \emph{local
routing}.

Together with the pair
$(u,v)$, we store in $\rho(r)$ the \emph{middle site} $m$  on the
shortest path from
$r$ to $\sigma(v)$, i.e., the vertex ``halfway'' between $r$ and
$\sigma(v)$.
Once we find $m$, we store $t$ at the header, and we
recursively route the packet from $r$ to $m$. When the packet
reaches $m$, we retrieve  $t$ from the header and continue the
routing
from $m$ to $t$. To keep track of intermediate targets during the
recursion,
we store a stack of targets in the header.
We call this second step that includes the
recursive routing through the middle site, the
\emph{global routing}.

\begin{figure}[htb]
 \centering
 \includegraphics[scale=0.8]{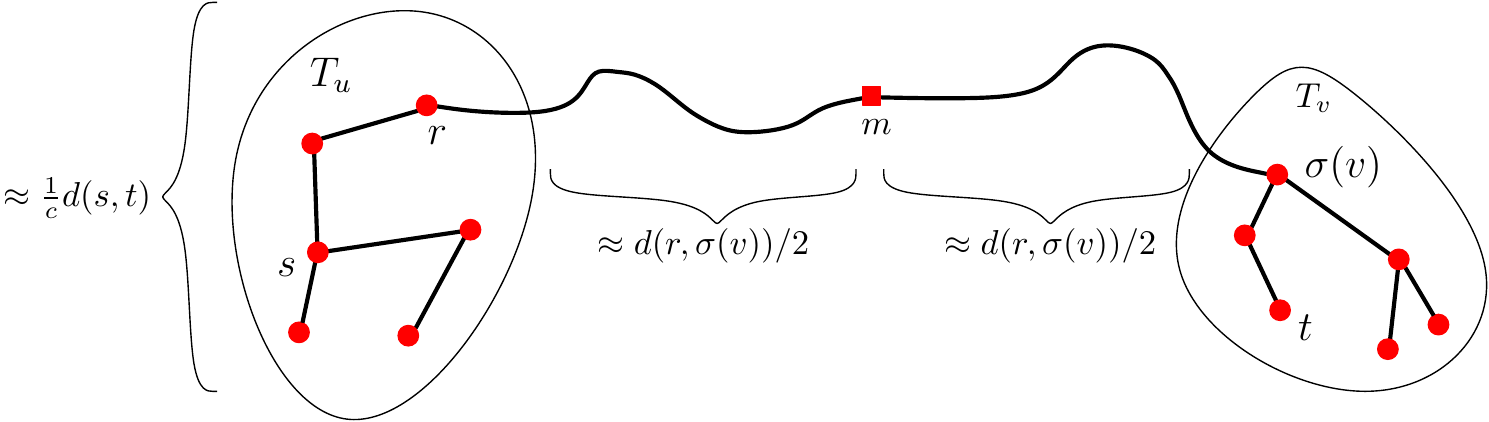}
 \caption{To route a packet from $s$ to $t$, we first walk along
 $T_u$ until
 we find $r$. Then we recursively route from $r$ to $m$ and from $m$
 to $t$.}
 \label{fig:recursive}
\end{figure}

We now describe our routing scheme in detail.
Let $1 + \eps$, $\eps > 0$, be the desired stretch
factor.

\subsection{Preprocessing}

The preprocessing phase works as follows.
We set $c = (\alpha/\eps)\log D$, where $D = \diam(S)$ and
$\alpha$ is a sufficiently large constant that we will
fix later. Then we
compute a
$c$-WSPD for $\UD(S)$.
As explained in Section~\ref{sec:wspd}, the WSPD consists
of a bounded degree spanning tree $T$ of $\UD(S)$,
a hierarchical balanced decomposition $H$ of $T$ whose
nodes $u \in H$ correspond to subtrees $T_u$ of $T$, and
a sequence $\Xi = (u_1, v_1), (u_2, v_2), \dots, (u_k, v_k)$
of $k = O(\vartheta c^2 n \log n) = O(\vartheta\eps^{-2}n \log n \log^2 D)$
well-separated
pairs that represent a partition of  $S \times S$.

First, we determine the \emph{labeling} $\ell$ for the sites in $S$.
This is done  as in the ``interval routing scheme''
of Santoro and Khatib~\cite{SantoroKhatibK85} for trees.
We perform a postorder traversal of $H$. Let $l$
be a counter which is initialized to $1$. Whenever
we encounter a leaf of $H$, we set the label $\ell(s)$
of the corresponding site $s \in S$ to $l$, and we increment
$l$ by $1$. Whenever we visit an internal node $u$ of $H$
for the last time, we annotate it with the interval $I_u$ of the
labels in $T_u$. Thus, a site $s \in S$ lies in
a subtree $T_u$ if and only if $\ell(s) \in I_u$.
Each label has $O(\log n)$ bits.

Next, we describe the routing tables.
Each routing table consists of two parts, the
\emph{local} routing table and the \emph{global} routing table.
The local routing table $\rho_L(s)$ of a site $s$ stores
the neighbors of $s$ in $T$, in counterclockwise order,
together with the levels in $H$ of the corresponding edges
(cf. Section~\ref{sec:wspd}).
Since $T$ has degree at most $6$,
each local routing table consists of $O(\log n)$
bits.

The global routing table  $\rho_G(s)$ of a site $s$ is obtained
as follows: we go through all $O(\log n)$ nodes $u$ of $H$ that
contain $s$ in their subtree $T_u$. By Lemma~\ref{lem:wspdpairs},
$\Xi$ contains at most $O(\vartheta c^2 |S_u|)$ well-separated pairs in which
$u$ represents one of the sets. We assign $O(\vartheta c^2) =
O(\vartheta \eps^{-2}\log^2 D)$ of these pairs to $s$, such that each
pair is assigned to exactly one site in $S_u$.
For each pair $(u, v)$ assigned to $s$, we store the interval $I_v$
corresponding to $S_v$. Furthermore, if $\sigma(v)$ is not
a neighbor of $s$, we store at $s$, together with the pair $(u, v)$, the
label $\ell(m)$ of the \emph{middle site} $m$ on a shortest
path $\pi$ from $s$ to $\sigma(v)$.
Formally, $m$ is a site on $\pi$ that minimizes the maximum
distance,  $\max\{d(s, m), d(m, \sigma(v))\}$,
to the endpoints of $\pi$.

\begin{lemma} \label{lem:table-size}
For every site $s$, $\rho_G(s)$ is of size
$O(\vartheta \eps^{-2}\log^2 n\log^2 D)$ bits.
\end{lemma}
\begin{proof}
A site $s$ lies in $O(\log n)$ different sets $S_u$, at most one for each
level of $H$. For each such set, we store $O(\vartheta\eps^{-2}\log^2 D)$ pairs in
$\rho_G(s)$, each of
which requires $O(\log n)$ bits.
\end{proof}

Finally, we argue that the routing scheme can be computed efficiently.
Our preprocessing algorithm proceeds in a centralized fashion and
processes the whole graph to determine the routing table for each node.
\begin{lemma}
\label{lem:preprocessingtime}
The preprocessing time for the routing scheme described above is
$O(n^2\log n + \vartheta n^2 +  \vartheta \eps^{-2}n\log n \log^2 D)$.
\end{lemma}
\begin{proof}
The $c$-WSPD can be computed in
$O(\vartheta  c^2 n \log n) =
O(\vartheta \eps^{-2} n \log n \log^2 D)$ time~\cite{GaoZh05}.
Within the same time bound, we can
distribute the WSPD-pairs to the sites in $S$ and compute the labels
for $S$.

It remains to compute the middle sites; we do this for all
pairs $(s,t) \in S\times S$ as follows: we first compute $\UD(S)$
explicitly.
Since $S$ has density $\vartheta$, we
have $O(\vartheta n)$ edges in $\UD(S)$,
and we can compute it naively in time $O(n^2)$.
For each $s \in S$, we compute
the shortest path tree $\mathcal{T}$ with root $s$. This takes total
time $O(n^2 \log n + \vartheta n^2)$, using $n$ invocations of
Dijkstra's algorithm.

\begin{figure}[htb]
\centering
\includegraphics{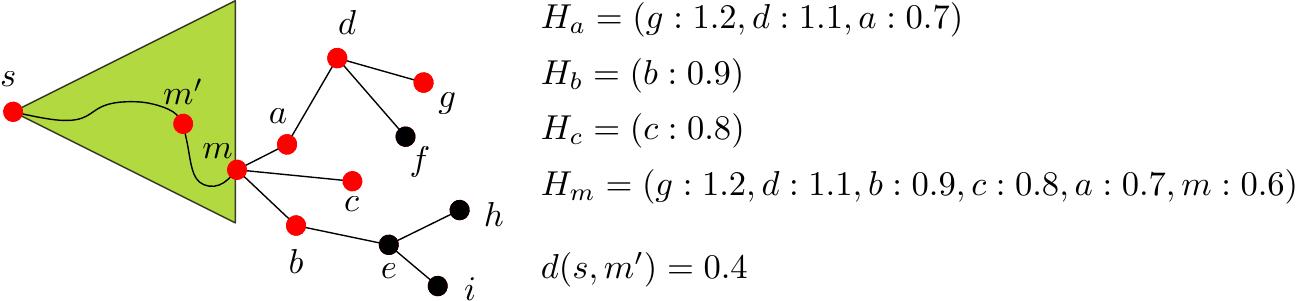}
\caption{$m$ is the middle site for $g$ and $d$.
$m'$ is a better middle site for vertices at least as close
to $m$ as $b$.}
\label{fig:preprocessing}
\end{figure}

For each $s \in S$, we perform a post-order
traversal of the shortest path tree  $\mathcal{T}$
to find the middle sites for all $s$-$t$-paths.
First, for each leaf $t$ of $\mathcal{T}$, we create a max-heap
that contains $t$ with $d(s,t)$ as the key.
We now describe how to process a site $m$ during the traversal.
First, we merge the heaps of all children of $m$ into a new heap $H_m$
and we insert $m$ into $H$ with $d(s, m)$ as key, see
Figure~\ref{fig:preprocessing}.
During the traversal, we maintain the invariant that $H_m$ contains all
sites
that are descendants of $m$ in $\mathcal{T}$ for which we have not yet
found a middle site. Furthermore, since $d(s,t)$ increases monotonically
along
every root-leaf path in $\mathcal{T}$, the sites for
which $m$ might be the middle site are a prefix of the decreasingly
sorted distances $d(s,t)$ with $t \in H$. Thus, to find the sites in
$H_m$ for which $m$ is the middle site, we repeatedly perform an
extract-max operation on $H_m$ to obtain the next candidate $t$. Then,  we
compare the value of
$\max \{ d(s,m),d(m,t)\}$ with $\max\{d(s,m'),d(m',t)\}$, where $m'$ is the
parent of $m$ in $\mathcal{T}$. That is, we check if $m'$ is a ``better''
middle site than $m$. If not, $m$ must be the middle site for $s$-$t$.
Otherwise, $m$ cannot be the middle site for any other site in $H_m$, and
we
proceed with our traversal.
Using, e.g., Fibonacci Heaps, we can merge two heaps in $O(1)$ time and
perform an extract-max operation in $O(\log n)$ amortized
time~\cite{Cormen09}.
Since each element of $\mathcal{T}$ is inserted and extracted at most
once, we need $O(n\log n)$
time to find the middle sites for $s$.
Thus, we can find all middle sites in time $O(n^2\log n)$ and
the total
preprocessing time is
$O(n^2\log n + \vartheta n^2 +  \vartheta \eps^{-2}n\log n \log^2 D)$.
\end{proof}

\subsection{Routing a Packet}
\label{sec:routing}

Suppose we are given two sites $s$ and $t$,
and we would like to route a packet from $s$ to $t$.
Recall our overall strategy:
we first perform a local exploration of $\UD(S)$ in
order to discover a site $r$ that stores a pair $(u, v) \in \Xi$
representing $(s,t)$ in its global routing table $\rho_G(r)$.
To find $r$, we consider the subtrees
of $T$ that contain $s$ by increasing size,
and we perform an Euler tour in each subtree until we find $r$.
In $\rho_G(r)$ we have stored the middle site $m$ of a shortest path
from $r$ to $\sigma(v)$.
We put $t$ into the header,
and we recursively route
the packet from $r$ to $m$. Once we reach $m$, we retrieve the original
target $t$ from the header and
recursively route from $m$ to $t$, see Algorithm~\ref{alg:routing}
for pseudo-code.

\paragraph{Local Routing: The Euler-Tour.}
We start at $s$, and we would like to find the site $r$ that stores
the pair
$(u,v)$ representing $(s,t)$.
By construction, both $s$ and
$r$ are contained in $S_u$, and it suffices to perform an Euler tour
on $T_u$ to discover $r$.
Since we do not know $u$ in advance, we begin with the parent of
the leaf in $H$ that contains $s$, and we explore all
nodes on the path to the root until we find $u$ (see
Figure~\ref{fig:local-routing}).

\begin{figure}[htb]
\centering
\includegraphics[scale=0.8]{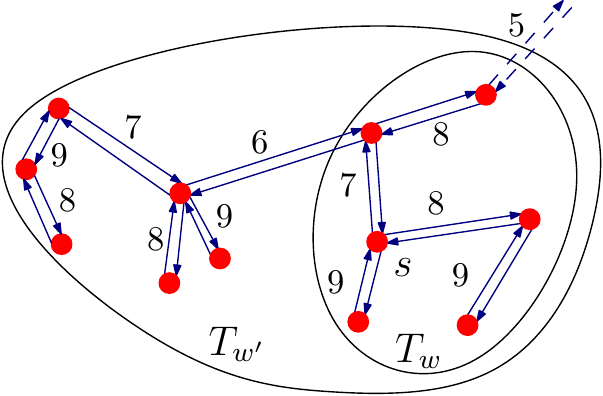}
 \caption{To find $r$ we do an Euler Tour on $T_w$, the subtree
 that contains $s$ whose edges have level at least $7$. Since we do not
 find $r$, we search the next larger subtree $T_{w'}$, where $w'$ is the
 parent of $w$ in $H$ by decreasing the search level to $6$.}
 \label{fig:local-routing}
\end{figure}

Let $w \in H$ be the node to be explored,
and let $l = \delta(w)$ be the depth of $w$ in $H$.
The header $h$ contains
$l$,
the current level being explored,
and $e$,
the (directed) \emph{start edge}.
The start edge $e$ is the unique edge of level $\delta(s) - 1$
incident to $s$, i.e., the edge associated with the parent of $s$ in $H$.
Recall that $T_w$ is a connected
component of the forest induced by all edges
of level at least $l$. We perform an Euler tour
on $T_w$ using the local routing tables.

To begin the tour, we follow the start edge $e$ from $s$.
This edge is contained in all non-trivial subtrees containing $s$.
Throughout the search, we maintain in the header the  previous vertex
before the current vertex, i.e.,
when we reach a vertex $r$  through the edge $(r',r)$, the header contains the
previous vertex $r'$.

Every time we visit a site $r$, we check for all WSPD-pairs $(u,v)$ in
$\rho_G(r)$ whether $\ell(t) \in I_v$, i.e., whether $t \in S_v$. If so,
we clear the local routing information from $h$, and we proceed
with the global routing.
If not, we retrieve from the header the vertex $r'$ through
which we just arrived to $r$ and we scan $\rho_L(r)$ to find
the first edge $(r,r'')$ of level at least $l$ in clockwise order after $(r, r')$,
going back to the beginning of
$\rho_L(r)$ if necessary.
If the edge $(r, r'')$ is different from the start edge $e$,
then $(r,r'')$ is next edge of the Euler tour. We remember
$r$ as the previous vertex in the
header, and we proceed to $r''$.
Otherwise, if $(r, r'') = e$, the Euler tour of $T_w$ tries to
traverse the start edge for a second time.
This means that $T_w$ does not contain the desired middle site.
We decrease $l$ by one, and we again follow the start edge.
Decreasing $l$ corresponds to proceeding to the parent of $w$ in $H$ and hence
to the next larger subtree.

\paragraph{Global Routing: The WSPD.}
Suppose we are at a site $s$ such that $\rho_G(s)$ contains the pair
$(u,v)$ with the target $t$ being in $S_v$.
If $t$ is not a neighbor
of $s$, then $\rho_G(s)$ must contain the label of a middle site $m$
for $(u,v)$.\footnote{By Lemma~\ref{lem:closepairs} if $t$ is not a neighbor of $s$ then $\sigma(v)$
cannot be a neighbor of $s$, and therefore $m$ must exist.}
We push
(the label of) $t$ onto the header stack,  we set
$m$ to be the new target and we apply the routing function again
(with  $\ell(m)$ and the new header).
If $t$ is a neighbor of $s$, we go directly to $t$ without changing
the target and the header.

When we reach the current destination and the stack at the header is not
empty then we pop the next target from the header and apply the routing
function again. Otherwise, the current destination is the final one and
the routing is complete.

\begin{algorithm}[htb]
\DontPrintSemicolon
\KwIn{currentSite $s$, targetLabel $\ell(t)$, header $h$}
\KwOut{nextSite, nextTargetLabel, header}
  \If(\tcc*[f]{intermediate target reached?}){$\ell(s) = \ell(t)$} {
      \If(\tcc*[f]{final target?}){$h.\textnormal{stack} = \emptyset$} {
    \Return $(s,  \perp, \perp)$
      }
      \Else {
    \Return $(s, h.\text{stack.pop}(), h)$
\label{line:pop}
      }
  } \ElseIf(\tcc*[f]{global routing}){$\rho(s)$ \textnormal{stores a WSPD-pair}
$(u,v)$ \textnormal{with} $t \in S_v$}{
    $h.\text{startEdge} \gets \emptyset$\;
    \If{$s \textnormal{ and }  t \textnormal{ are neighbors in }
      \UD(S)$}{
    \Return $(t, \ell(t), h)$
    }
    \Else{
      $\text{nextTargetLabel} \gets \text{label of middle site
      for } (u,v)$\;
    $h.\text{stack.push}(\ell(t))$\;
        \Return $(s, \text{nextTargetLabel}, h)$
    }
  }
  \Else(\tcc*[f]{local routing}){
    \If(\tcc*[f]{initialize local routing}){h.\textnormal{startEdge} $  = \emptyset$}{
      $h.\text{level} \gets \delta(s) - 1$\;
      $r \gets$ neighbor of $s$ with level of edge
        $sr = h.\text{level}$\;
      $h.\text{startEdge} \gets sr$\;
    } \Else{
      $r \gets$ clockwise neighbor of $s$ following $h.\text{prevVertex}$ with level of
         $sr \geq h.\text{level}$\;
        \If(\tcc*[f]{Euler tour has finished tree}){$sr = h.\text{startEdge}$}{
          $h.\text{level} \gets h.\text{level} - 1$
          \tcc*{try next level}
        }
    }
    $h.\text{prevVertex} \gets s$\;
    \Return $(r,\ell(t),h)$
  }
\caption{The routing function.
In the resutling triple the source is either the current
vertex $s$, in which case we apply the routing function again at $s$, or
it is a neighbor $s'$ of $s$, in which case the packet is sent to $s'$.}
\label{alg:routing}
\end{algorithm}

\subsection{Analysis of the Routing Scheme}
We now prove that the described routing scheme is correct and has low
stretch, i.e., that
for any two sites $s$ and $t$, it produces a routing path $s = s_0, \dots, s_k = t$
of length at most $(1+\eps)d(s,t)$.

\subsubsection{Correctness}
We now prove that our scheme is correct.
For small distances, the routing path is actually
optimal.

\begin{lemma}
\label{lem:correctness}
Let $s, t$ be two sites in $S$. Then, the routing
scheme produces a routing path $s_0, s_1, \dots, s_k$ with the following
properties
\begin{enumerate}[(i)]
 \item  $s_0 = s$ and $s_k = t$,
 \item  the header stack is in the same
state at the beginning and at the end of the routing path, and
\item if $d(s,t) < c$, then $d_\rho (s,t) = d(s,t)$.
\end{enumerate}
\end{lemma}
\begin{proof}
 We prove that our routing scheme has properties (i)--(iii)
by induction on the rank of
$d(s, t)$ in the sorted list of the pairwise
distances in $\UD(S)$.

For the base case, consider the edges $st$ in $\UD(G)$, i.e.,
$d(s,t) = |st| \leq 1$. By Lemma~\ref{lem:closepairs},
there exists a pair $(u,v) \in \Xi$ with
$S_u = \{s\}$ and $S_v = \{t\}$. Thus,
Algorithm~\ref{alg:routing} correctly routes to $t$ in one step.
It uses a shortest path and does not manipulate the header stack.
All properties are fulfilled.

Next, consider an arbitrary pair $s,t$ with $1 < d(s,t) < c$.
By Lemma~\ref{lem:closepairs}, there is a pair $(u,v) \in \Xi$ with
$S_u = \{s\}$ and
$S_v = \{t\}$. By construction, $(u,v)$ is stored in $\rho_G(s)$ and
the routing algorithm directly proceeds to the global routing phase.
Since $d(s, t) > 1$, the routing table contains a middle site
$m$ and since $S_u$ and $S_v$ are singletons, $m$ is a
middle site on a shortest path from $s$ to $t$.
Algorithm~\ref{alg:routing} pushes $\ell(t)$ onto the stack and sets $m$
as the
new target. By induction, the routing scheme now routes the packet along
a shortest path from $s$ to $m$ (items i and iii of the induction hypothesis), and when the packet arrives at
$m$, the target label $\ell(t)$ is at the top of the stack (item ii).
Thus, Algorithm~\ref{alg:routing} executes line~\ref{line:pop},
and routes the packet from $m$ to $t$. Again by induction,
the packet now follows a shortest path from $m$ to $t$ (i, iii), and
when the packet arrives at $t$, the stack is in the same state
a before pushing $\ell(t)$ (ii).
The claim follows.

Finally, consider an arbitrary pair $s,t \in S$ such that $d(s, t) \ge c$.
By construction, our routing scheme will eventually find
a site $r \in S$ whose global routing table stores a WSPD-pair $(u, v)$
that represents $(s, t)$. Up to the point in which we reach $r$
the stack remains unchanged (see Figure~\ref{fig:stack}).
If $\sigma(v)$ is a neighbor of $r$ then by
Lemma~\ref{lem:diameter}, $|S_u| = |S_v| = 1$. So
$\sigma(v) = t$ and the packet arrives to $t$ from $r$ in a single step
with the header in its original state.
Otherwise, there is a middle site $m$
associated with $(u, v)$ in $\rho(r)$.

Algorithm~\ref{alg:routing} pushes $\ell(t)$ onto the stack and
sets $m$ as the
new target. By induction, the routing scheme routes the packet correctly
from $s$ to $m$ (i), and when the packet arrives at
$m$, the target label $\ell(t)$ is at the top of the stack (ii).
Thus, Algorithm~\ref{alg:routing} executes line~\ref{line:pop},
and routes the packet from $m$ to $t$. Again by induction,
the packet arrives at $t$, with the stack in the same state as before pushing
$\ell(t)$ (i, ii).
\end{proof}

\begin{figure}[htb]
\centering
\includegraphics[scale=0.9]{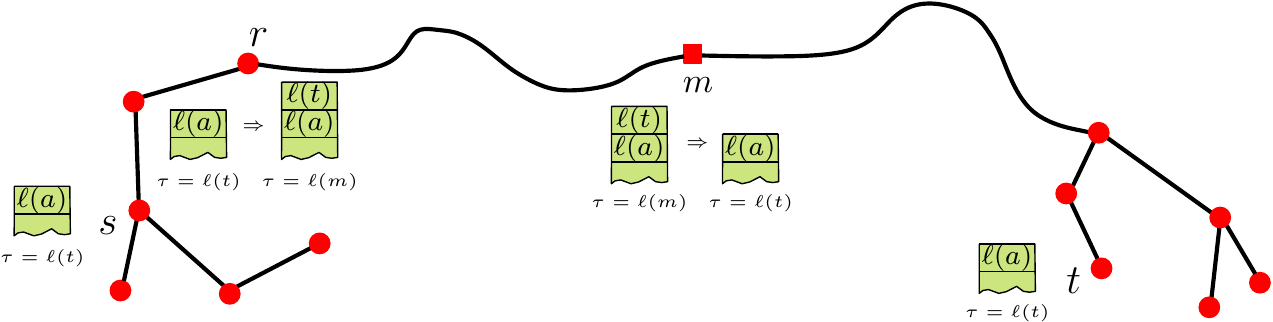}
\caption{How the stack (green) and the target label $\tau$ change due to the
global routing.}
\label{fig:stack}
\end{figure}

\subsubsection{Stretch factor}
The analysis of the stretch factor requires some more technical
work. We begin with a lemma that justifies the term
``middle site''.

\begin{lemma}
\label{lem:middlesite}
Let $s,t$ be two sites in $S$ with $d(s, t) \geq c \geq 13$
and let $(u,v) \in \Xi$ be the WSPD-pair that represents
$(s,t)$.
If $m$ is a middle
site of a shortest path from $s$ to $\sigma(v)$ in $\UD(S)$, then
\begin{enumerate}[(i)]
 \item $d(s,m) + d(m,t) \leq \left(1+\frac{2}{c}\right)d(s,t)$, and
 \item $d(s,m),d(m,t) \leq \frac{5}{8} d(s,t)$.
\end{enumerate}
\end{lemma}
\begin{proof}
For (i) we have
\begin{align*}
 d(s,m) + d(m,t) &\leq d(s,m) + d(m,\sigma(v)) + d(\sigma(v),t)&&\text{(triangle inequality)}\\
 &= d(s,\sigma(v)) + d(\sigma(v),t) &&\text{($m$ is middle site)}\\
 & \leq d(s,t) + 2d(\sigma(v),t)&& \text{(triangle inequality)}\\
&\leq \left(1 + \frac{2}{c}\right)d(s,t),&& \text{(Lemma~\ref{lem:diameter})}
\end{align*}
 where the last inequality also uses the fact
that $d(\sigma(v),t) \leq \diam(S_v)$.

For (ii) let $\pi$ be a shortest path from $s$ to $\sigma(v)$ that contains
$m$, and let $m'$ be the point on $\pi$ with distance
$d(s, \sigma(v))/2$ from $s$ and from $\sigma(v)$ ($m'$ may lie on an
edge of $\pi$). Since the edges of $\pi$ have length at most
$1$ and since $m$ is the middle site, we have
$d(m, m') = |mm'| \leq 1/2$. Hence,
\begin{equation}
\label{equ:max}
\max \{d(s, m), d(m, \sigma(v))\}  \leq  \frac{d(s, \sigma(v))}{2} + \frac{1}{2}.
\end{equation}
Using triangle inequality and Lemma~\ref{lem:diameter} we get
\begin{equation}\label{equ:dmmt}
  \left(1-\frac{1}{c}\right)d(m,t) \leq d(m,\sigma(v))
\end{equation}
and
\begin{equation}\label{equ:dssv}
d(s,\sigma(v)) \leq \left(1+\frac{1}{c}\right)d(s,t).
\end{equation}
Using (\ref{equ:max}), (\ref{equ:dmmt}), and (\ref{equ:dssv}) we can derive
\begin{align*}
 \max \{d(s,m),d(m,t)\} &
 \leq \left(1+\frac{1}{c-1}\right) \max\{d(s,m),d(m,\sigma(v))\}&& \text{(by (\ref{equ:dmmt}))}\\
 & \leq \left(1+\frac{1}{c-1}\right)\left(\frac{d(s, \sigma(v))}{2} + \frac{1}{2}\right)&& \text{(by (\ref{equ:max}))}\\
 & \leq \left(1+\frac{1}{c-1}\right)\left(\left(1+\frac{1}{c}\right)\frac{d(s, t)}{2} + \frac{1}{2}\right) && \text{(by (\ref{equ:dssv}))}\\
 & \leq \left(1+\frac{1}{c-1}\right)\left(1+\frac{2}{c}\right)\frac{d(s,t)}{2}&& \text{(*)} \\
 & \leq \left(\frac{1}{2}+\frac{3}{2c-2}\right)d(s,t),&& \\
\end{align*}
where (*) is due to the assumption that $d(s,t) \geq c$. Now (ii) follows
from the assumption that $c \geq 13$.
\end{proof}
In the next lemma, we bound the distance traveled during the local routing.

\begin{lemma}
\label{lem:localroutingdistance}
Let $s,t$ be two sites in $S$ with $d(s,t) \geq c$.
Then, the total distance traveled by the packet during the local routing
phase before the
WSPD-pair representing $(s,t)$ is discovered, is at most
$(48/c) d(s,t)$.
\end{lemma}
\begin{proof}
Let $(u, v)$ be the WSPD-pair representing $(s, t)$, and
let $u_0, u_1, \dots, u_k = u$ be the path in $H$
from the leaf $u_0$ of $s$ to $u$.
Let $T_0, T_1, \dots, T_k$ and
$S_0, S_1, \dots, S_k$ be the corresponding
subtrees of $T$ and sites of $S$. The local routing algorithm iteratively
performs an Euler tour of $T_0, T_1, \dots, T_k$ (the tour of $T_k$ may stop
early). An Euler tour in
$T_i$ takes $2|S_i| - 2$ steps, and each edge has length at most $1$.
As described in Section~\ref{sec:wspd}, for $i = 0. \dots, k-1$, the WSPD ensures that
\[
|S_{i}| \leq |S_{i+1}| - \left\lceil \frac{|S_{i+1}| - 1}{6}\right\rceil
 \leq \frac{5}{6}\,|S_{i+1}| +  \frac{1}{6}
  \leq \frac{11}{12}\,|S_{i+1}|,
\]
since $|S_{i+1}| \geq 2$. It follows that the total distance for the local routing is
at most
\[
  \sum_{i=0}^{k} (2|S_i| - 2) \leq 2 |S_k| \sum_{i=0}^k \left(\frac{11}{12}\right)^i
\leq 24 |S_k|.
\]
By Lemma~\ref{lem:diameter},  we have $d(s,t) \geq c(|S_u| -1)$ and
since $S_k = S_u$ the total distance is bounded by
$24 |S_u| \leq 24 (d(s,t)/c + 1) \leq (48/c) d(s,t)$,
where the last inequality is true for $d(s,t) \geq c$.
\end{proof}

Finally, we can bound the stretch factor as follows.
\begin{lemma}
\label{lem:stretch}
For any two sites $s$ and $t$, we have
$d_\rho(s, t) \leq \bigl(1 + \eps) d(s, t)$.
\end{lemma}

\begin{proof}
We show by induction on $d_\rho(s, t)$
that there is an $\alpha > 0$ with
\begin{equation*}
d_\rho(s, t) \leq \left(1 + \frac{\alpha}{c}\,\log d(s, t) \right) d(s, t).
\end{equation*}
Since $d(s,t) \leq \diam(S) = D$, the lemma then follows from our
choice of $c =
(\alpha/\eps) \log D$.

If $d(s, t) < c$, the claim follows by
Lemma~\ref{lem:correctness}(iii): the packet is
routed along a shortest path and incurs no detour.
If $d(s, t) \geq c$, Algorithm~\ref{alg:routing} performs
a local routing to find the site $r$ that has the WSPD-pair $(u,v)$
representing $(s,t)$ stored in $\rho_G(r)$. Then the packet is routed
recursively from $r$ to the middle site $m$ and from $m$ to $t$.
By Lemma~\ref{lem:localroutingdistance} the length of
the routing path is
$d_\rho(s, t)  \leq (48/c) d(s, t)
+ d_\rho(r, m) + d_\rho(m, t)$,
and by induction we get
\begin{align*}
 d_\rho(s, t) & \leq \frac{48}{c}\, d(s, t) +
  \left(1 + \frac{\alpha}{c}\, \log d(r, m) \right) d(r, m)
   + \left(1 + \frac{\alpha}{c} \,\log d(m, t)
\right) d(m,t).
\intertext{Since $m$ is a middle site on a
shortest $r$-$\sigma(v)$-path in $\UD(S)$, Lemma~\ref{lem:middlesite}(i),(ii)
and the fact that $\log(5/8) \leq - 1/2$ imply}
 d_\rho(s, t) & \leq \frac{48}{c} d(s, t) +
  \left(1 + \frac{\alpha}{c}\, \log(d(r, t))
    -\frac{\alpha}{2c}\right)\left(1+\frac{2}{c}\right)d(r, t).
\end{align*}
By the triangle inequality we have $d(r,t) \leq d(s,t) +
\diam(S_u)$, so Lemma~\ref{lem:diameter} gives
\begin{align*}
d_\rho(s, t)
& \leq \frac{48}{c}\, d(s, t)
  + \left(1 + \frac{\alpha}{c}\log\left(\left(1+\frac{1}{c}\right)d(s,t)\right)
    -\frac{\alpha}{2c}\right)\left(1+\frac{2}{c}\right)\left(1+\frac{1}{c}\right)d(s, t)  \\
& \leq \frac{48}{c}\, d(s, t)
  + \left(1 + \frac{\alpha}{c}\log\left(\left(1+\frac{1}{c}\right)d(s,t)\right)
    -\frac{\alpha}{2c}\right)\left(1+\frac{4}{c}\right)d(s, t),
\end{align*}
for $c$ large enough.
For $\alpha > 192$, we can eliminate the first term to get
\begin{align*}
 d_\rho(s, t) & \leq \left(1 + \frac{\alpha}{c} \log\left(\left(1+\frac{1}{c}\right)d(s,t)\right)
    - \frac{\alpha}{4c}\right)\left(1+\frac{4}{c}\right)d(s, t),
\intertext{and since now $c \geq 192$ and hence $\log (1+1/c) \leq 1/8$,
}
 d_\rho(s, t) &\leq \left(1 + \frac{\alpha}{c} \log(d(s,t))
    -\frac{\alpha}{8c}\right)\left(1+\frac{4}{c}\right)d(s, t)
  = \left(1 + \frac{\alpha}{c}\log d(s,t)\right) d(s, t) + \Delta,
\end{align*}
with
\[
\Delta = - \frac{\alpha}{8c}\,  \left(1+\frac{4}{c}\right)d(s,t)
+ \frac{4}{c}\,d(s,t)\left(1 + \frac{\alpha}{c}\, \log d(s,t)
    \right).
\]
It remains to show that $\Delta \leq  0$, i.e., that
\begin{align*}
    \frac{4}{c}d(s,t)\left(1 + \frac{\alpha}{c}\, \log d(s,t) \right)
\leq \frac{\alpha}{8c}\,\left(1+\frac{4}{c}\right)d(s,t).
\end{align*}
Now, since we picked $c = (\alpha/\eps)\log D$ and $\alpha \geq 192$,
we have
\[
1 + \frac{\alpha}{c}\, \log(d(s,t)) \leq 2
\leq
\frac{\alpha}{32} \left(1 + \frac{4}{c}\right),
\]
as desired.
This finishes the proof.
\end{proof}

Combining Lemma~\ref{lem:table-size}, Lemma~\ref{lem:preprocessingtime},
and Lemma~\ref{lem:stretch}
we obtain
the following theorem.
\begin{theorem}
Let $S$ be a set of $n$ sites in the plane with density $\vartheta$.
For any $\eps > 0$, we can preprocess $S$ into a routing scheme for
$\UD(S)$ with
labels of size $O(\log n)$ bits and routing
 tables of size $O(\vartheta\eps^{-2}\log^2 n \log^2 D)$ bits, where $D$ is the
diameter of $\UD(S)$. For any two sites $s$,$t$, the scheme produces a
routing path with $d_\rho(s,t) \leq (1+\eps) d(s,t)$ and during the routing
the maximum header size is $O(\log n \log D)$. The preprocessing time is
$O(n^2\log n + \vartheta n^2 +  \vartheta \eps^{-2}n\log n \log^2 D)$.
\end{theorem}

\subsection{Extension to Arbitrary Density}
\label{sec:unboundeddensity}

Let $1+\eps$, $\eps > 0$, be the desired stretch factor.
To extend the routing scheme to point sets of
unbounded density, we follow a strategy similar
to Gao and Zhang~\cite[Section~4.2]{GaoZh05}:
we first pick an appropriate $\eps_1 > 0$, and we
compute an \emph{$\eps_1$-net} $R \subseteq S$,
i.e., a subset $R$ of sites such that each site in $S$ has
distance at most $\eps_1$ to the closest site in $R$ and
such that any two sites in $R$ are at distance at least
$\eps_1$ from each other, see Figure~\ref{fig:density}.
\begin{figure}[htb]
\centering
\includegraphics[scale=0.8]{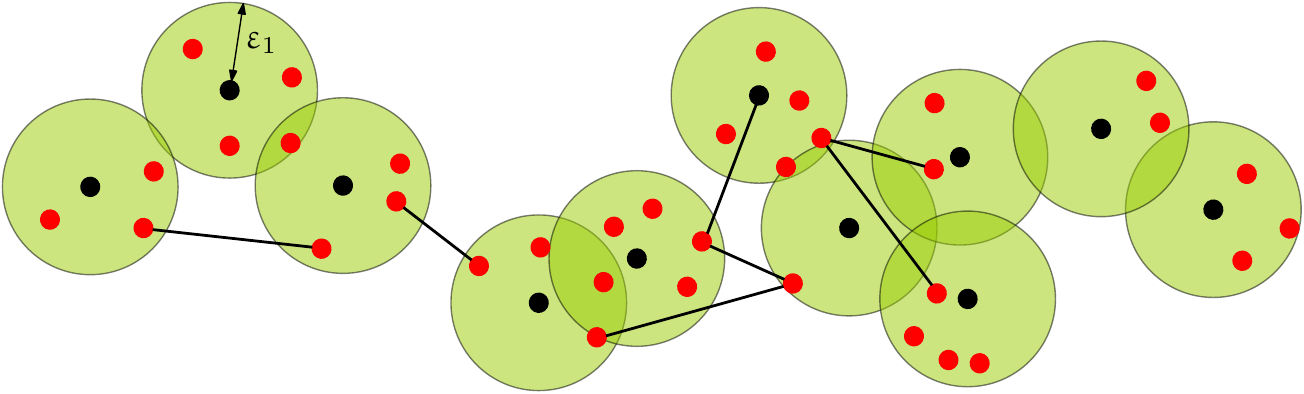}
\caption{The set $R$ (black) and the bridges (endpoints of black edges) form
the set $Z$.}
\label{fig:density}
\end{figure}

As we show below, it is easy to see that $R$ has density $O(\eps_1^{-2})$,
and we would like to represent each site in $S$ by the
closest site in $R$. However, the connected components of $\UD(R)$ might differ
from those of $\UD(S)$. To rectify this, we add additional sites
to $R$. This is done as follows:
two sites $s, t \in R$ are called \emph{neighbors} if
$|st| > 1$, but there are $p, q \in S$ such that
$s, p, q, t$ is a path in $\UD(S)$ and such that
$|sp|  \leq \eps_1$ and $|qt| \leq \eps_1$
(possibly, $s=p$ or $q = t$).
In this case, the pair of sites $p$ and $q$ is called a \emph{bridge}
for $s, t$. Let $R'$ be a point set that contains
an arbitrary bridge
for each pair of neighbors in $R$.
Set $Z = R \cup R'$. The following simple volume argument
shows that $Z$ has density $\vartheta = O(\eps_1^{-3})$.
\begin{lemma}
\label{lem:r-densitybound}
The set $R$ has density $O(\eps_1^{-2})$ and the set $Z$ has density
$O(\eps_1^{-3})$.
\end{lemma}
\begin{proof}
Let $D$ be a unit disk and let $D'$ be the disk with radius $1+\eps_1/2$
concentric
to $D$. For each $s \in R \cap D$, the disk $D(s,\eps_1/2)$
with center $s$ and radius $\eps_1/2$ is contained in $D'$.
Let $s,t \in R$ be two sites.
By construction $|st|\geq \eps_1$ and thus the disks
$D(s, \eps_1 / 2)$ and $D(t, \eps_1 / 2)$ are disjoint.
A disk of radius $\eps_1 / 2$ has area $\pi \eps_1^2/4$ and the area of
$D'$
is $\pi(1+\eps_1/2)^2$. Thus we can place $O(\eps_1^{-2})$ such disks into
$D'$ disjointly. Hence, the density of $R$ is
$O(\eps_1^{-2})$.

Now, let $D''$ be the disk
with radius $1 + \eps_1$
concentric to $D$.
To bound the density of $Z$,
let $p \in (Z \setminus R) \cap D$ be a site that belongs to a bridge.
Suppose that $s \in R$ is responsible for $p$ being a bridge site.
Then we have $|sp| \leq \eps_1$ and by
construction $s \in D''$. We charge $p$ to $s$.
The same volume argument as above shows that
$|R \cap D''| = O(1/\eps_1^{-2})$.
Below we show that $s$ has $O(\eps_1^{-1})$ neighbors and thus
can get $O(\eps_1^{-1})$ charges from bridge sites. Hence,
the number of bridge sites in $Z \cap D$, and also
the density of $Z$, is $O(\eps_1^{-3})$.

Consider the annulus $A$
around $s$ with inner radius
$1$ and outer radius $1+2\eps_1$.
All neighbors of $s$ must lie in $A$.
Let $A'$ the annulus concentric to $A$ with inner radius
$1 - \eps_1/2$ and outer radius $1+(5/2)\eps_1$.
The area of $A'$ is
\begin{align*}
 \pi\left(1+\frac{5}{2}\eps_1\right)^2 -  \pi\left(1 - \frac{\eps_1}{2}\right)^2 = 6\eps_1 + 6\eps_1^2,
\end{align*}
and thus, since $R$ is an $\eps_1$-net, we
can place $O(\eps_1^{-1})$ sites of $R$ in $A$.
Hence, $s$ has $O(\eps_1^{-1})$ neighbors, as claimed.
\end{proof}

Furthermore, Gao and Zhang show the following:
\begin{lemma}[Lemma~4.8 and Lemma~4.9 in \cite{GaoZh05}]
\label{lem:deltacover}
We can compute $Z$ in $O((n/\eps_1^2) \log n)$ time,
and if $d^Z(\cdot, \cdot)$ denotes the shortest
path distance in $\UD(Z)$, then, for any $s, t \in R$,
we have
$
  d^Z(s, t) \leq (1+ 12\eps_1) d(s, t) +  12\eps_1$.
\end{lemma}

Now, our extended routing scheme proceeds as follows:
first, we compute $R$ and $Z$ as described above, and
we perform the preprocessing algorithm for $Z$ with $\eps_1$
as the stretch parameter.
We assign arbitrary new labels to the sites in $S \setminus Z$.
Then, we extend the label $\ell(s)$ of each site $s \in S$,
such that it also contains the label of a site
in $R$ closest to $s$. The label size remains $O(\log n)$.

To route between two sites $s, t \in S$, we first check whether
we can go from $s$ to $t$ in one step (we assume that this can
be checked locally by the routing function). If so, we route
the packet directly. Otherwise, we have $d(s,t) > 1$.
Let $s', t' \in R$ be the closest sites in $R$ to $s$ and to $t$.
By construction,  we can obtain $s'$ and $t'$ from $\ell(s)$ and $\ell(t)$.
Now, we first go from $s$ to $s'$. Then,
we use the low-density algorithm to route from $s'$ to
$t'$ in $\UD(Z)$, and finally we go from $t'$ to $t$
in one step.
Using the discussion above, the total routing distance is bounded by
\begin{align*}
d_\rho(s,t) & \leq  |ss'| + d^Z_\rho(s', t') + |t't|,\\
\intertext{where $d^Z_\rho(\cdot,\cdot)$ is the routing distance in
$\UD(Z)$. By Lemma~\ref{lem:stretch} and \ref{lem:deltacover}, this is}
& \leq
\eps_1 + (1 + \eps_1)d^Z(s', t') + \eps_1\\
& \leq
2\eps_1 + (1 + \eps_1)\bigl((1+12\eps_1)d(s', t') + 12\eps_1\bigr), \\
\intertext{and by using the triangle inequality twice this is}
& \leq  2\eps_1 + (1 + \eps_1)\bigl((1+12\eps_1)(d(s, t) + 2\eps_1) +
12\eps_1\bigr).\\
\intertext{Rearranging and using $d(s,t) > 1$ yields }
& \leq  (1+29\eps_1 + 50\eps_1^2 + 24\eps_1^3) d(s, t) \leq
(1+\eps) d(s, t),
\end{align*}
where the last inequality holds for $\eps_1 \leq \eps/103$. This establishes
our main theorem:
\begin{theorem}
Let $S$ be a set of $n$ sites in the plane. For any $\eps
> 0$, we can preprocess $S$ into a routing scheme for $\UD(S)$ with labels
of $O(\log n)$ bits and routing
 tables of size $O(\eps^{-5}\log^2 n \log^2 D)$, where $D$ is the
diameter of $\UD(S)$. For any two sites $s$,$t$, the scheme produces a
routing path
with $d_\rho(s,t) \leq (1+\eps) d(s,t)$ and during the routing
the maximum header size is $O(\log n \log D)$. The preprocessing time is
$O(n^2\log n + \eps^{-3} n^2 +  \eps^{-5}n\log n \log^2 D)$.
\end{theorem}
\begin{proof}
The theorem follows from the above discussion and from the fact
that the set $Z$ has density $O(\eps^{-3})$, by our choice of
$\eps_1$.
\end{proof}

\section{Conclusion}
We have presented an efficient routing scheme for unit disk
graphs that produces a routing path whose length can be made
arbitrarily close to optimal. For this, we used the fact that
the unit disk graph metric admits a small WSPD. Our techniques
almost solely rely on properties of well-separated pairs and
thus we expect our approach to generalize to other graph metrics
for which WSPDs can be found. One such example is the
\emph{hop-distance} $d_h(\cdot,\cdot)$ in unit disk graphs,
in which all edges have length $1$. Let $S$ be a set of sites
and let $\diam_h(S)$ denote the diameter of $S$ in terms of
$d_h(\cdot,\cdot)$. Since $\diam_h(S) \leq |S|-1$ and
$|st| \leq d_h(s,t)$ for every two sites $s,t \in S$,
the well-separation condition (\ref{equ:well-separated})
implies also separation with respect to the
hop-distance.
Thus, we can also find a routing scheme that approximates the
number of hops used in the routing path instead of its Euclidean length.

Various open questions remain. First of all, it would be interesting
to improve the size of the routing tables. One way to
achieve this might be to decrease the dependency on $\eps$.
The $\eps^{-5}$-factor seems to be rather high. It is
mostly due to the $\eps^{-3}$-factor that we introduced in
Section~\ref{sec:unboundeddensity} when
extending the routing scheme to a set of sites of
 unbounded density.  Further improvements might be on the side
of the WSPD: traditional WSPDs have only $O(c^2 n)$ pairs, while the WSPD of
Gao and Zhang has an additional logarithmic factor. Whether
this factor can be avoided is still an open question and any
improvement in the number of pairs would immediately decrease
the size of our routing tables by the same
amount.

Furthermore, our routing scheme
makes extensive use of a modifiable header. While this is coherent
with the usual model for routing schemes,  the scheme of
Yan~\etal~does not need such a header.
In order to be completely comparable to their result, we would need to
have a routing scheme that only requires a small routing table
to produce a routing path with stretch $1+\eps$.

\bibliographystyle{abbrv}
\bibliography{literature}
\end{document}